\documentclass{article}

\usepackage[letterpaper,hmargin=1in,vmargin=1.25in]{geometry}
\usepackage{authblk}
\usepackage[utf8]{inputenc} 
\usepackage[T1]{fontenc}    
\usepackage{hyperref}       
\usepackage{url}            
\usepackage{amsfonts}       
\usepackage{nicefrac}       
\usepackage{microtype}      
\usepackage{graphicx}
\usepackage{doi}

\usepackage{bbm}
\usepackage[scr=esstix,cal=boondox]{mathalfa}

\usepackage{comment}
\newtheorem{theorem}{Theorem}
\newtheorem{lemma}[theorem]{Lemma}
\newtheorem{corollary}[theorem]{Corollary}

\def \sset {{\Scal}}

\usepackage{amssymb,amsmath,latexsym,xspace, bbm}
\usepackage{mathtools}

\newenvironment{proof}
{\textit{\bf Proof:} }
{$\square$\\ \\}

\newcommand{\given}{{\;|\;}}
\newcommand{\probb}[1]{{\mathbb{P}}\bigl(\,{#1}\,\bigr)}
\newcommand{\prob}{\mathbb{P}}
\newcommand{\E}{\mathbb{E}}
\def \L {{\mathcal{L}}}
\def \M {{{M}}}
\def\sfrac#1#2{{\textstyle{#1 \over #2}}}
\def \half {{{ \sfrac{1}{2}}}}

\def \Half {{{ \frac{1}{2}}}}
\def \Quarter {{{ \frac{1}{4}}}}
\def \C {{{C}}}
\def \th {{{\rm Th}}}
\def \calls {{\mathcal{C}}}
\def \length {{\rm length}}
\def \coinflip {{{\rm Bernoulli}(\half)}}
\def \advantage {\mathbf{MaxAdv}_{\mathbf{\numoraclecalls,\numqueries}}}
\def \e {{{\mathbb E }}}
\def \el {{{\mathbb E_\leakageval }}}
\def \bias {{\mathcal{B}}}
\def \leak {{\Psi}}

\def \leakagevals {{\mathcal{L}}}
\def \ent {{\mathbf{H}}}
\def \P {{\bf P}}

\def \probl {{\prob_\leakageval}}
\def \g {{g}}
\def \lam {{\lambda}}
\def \mur {{\mu_\leakageval^r}}
\def \reals {{\bf R}}
\def \T {{\mathcal T}}

\newcommand{\probp}[1]{{\mathbb{P}}\left({#1}\right)}

\newcommand{\Scal}{\mathcal{S}}

\newcommand{\Ecal}{\mathcal{E}}

\def \messagespace {\mathcal{M}} 
\def \messagelength {\mathcal{m}} 
\def \keyspace {\mathcal{K}} 
\def \keyvariable {K} 
\def \keyval {k} 
\def \keylength {\mathcal{k}} 
\def \numqueries {q} 
\def \leakageval {l} 
\def \leakagelength {\mathcal{l}} 
\def \leakagevariable {L} 
\def \numoraclecalls {\mathcal{r}} 
\def \totaltime {T} 
\def \numprobes {n} 
\def \s {\mathcal{s}} 
\def \leakagefunc {\Phi} 

\title{Format Preserving Encryption in the Bounded Retrieval Model}
\author[1]{Ben~Morris}
\author[1]{Hans~Oberschelp}
\author[1]{Hamilton Samraj~Santhakumar}
\affil[1]{Department of Mathematics, University of California, Davis}
\date{July 16, 2023}

\begin{document}

\maketitle

\begin{abstract}
  In the bounded retrieval model, the adversary can leak a certain amount of information from the message sender's computer (e.g., $10$ percent of the hard drive). Bellare, Kane and Rogaway give an efficient symmetric encryption scheme in the bounded retrieval model. Their scheme uses a giant key (a key so large only a fraction of it can be leaked.) One property of their scheme is that the encrypted message is larger than the original message. Rogaway asked if an efficient scheme exists that does not increase the size of the message. In this paper we present such a scheme. 
\end{abstract}


\section{Introduction}
The present paper attempts to solve the problem of format preserving encryption in the bounded retrieval model, by constructing a pseudorandom permutation and providing concrete security bounds in the random oracle model. The bounded retrieval model was introduced to study cryptographic protocols that remain secure in the presence of an adversary that can transmit or leak private information from the host's computer to a remote home base. One example of such an adversary is an APT (Advanced Persistent Threat), which is a malware that stays undetected in the host's network and tries to ex-filtrate the secret keys used by the host. The premise of the bounded retrieval model is that such an adversary cannot move a large amount of data to a remote base without being detected or that it can only communicate with the remote base through a very narrow channel. That is, the model assumes an upper bound on the amount of data that an adversary can leak. In \cite{big-key} Bellare, Kane and Rogaway introduce an efficient symmetric encryption scheme in this model and give concrete security bounds for it. They assume that the secret key is very large and model the leaked data as a function that takes the secret key as the input and outputs a smaller string. The length of this string is a parameter on which the security bounds depend. Their algorithm uses a random seed $R$ along with the big key to generate a key of conventional length that is indistinguishable from a random string of the same length even when the function used to model the leaked data depends on calls to the random oracle that the algorithm uses. It then uses this newly generated key and any of the conventionally available symmetric encryption schemes, say an AES mode of operation, to create a ciphertext $C$. Finally it outputs $(R,C)$.

The above scheme is not format preserving since the final ciphertext $(R,C)$ is longer than the original message $M$. A question posed by Phillip Rogaway (personal communication) is whether a secure format preserving encryption scheme exists in the bounded retrieval model. Another way to pose this question is as follows : If the adversary is allowed to leak data, is it possible to construct a pseudorandom permutation that is secure under some notion of security, say the CCA notion of security? The aim of this paper is to answer this question. Unfortunately it is not possible to come up with a pseudorandom permutation that is secure under the strong notion of CCA security. This is because in the CCA model, before trying to distinguish between a random permutation and the pseudorandom permutation, the adversary can choose to look at a sequence of plaintext-ciphertext pairs that he chooses. If a leakage of data is allowed, the adversary can simply leak a single plaintext-ciphertext pair and use it to gain a very high CCA advantage. Hence we weaken the notion of security by requiring that the adversary can only look at a sequence of plaintext-ciphertext pairs where the plaintexts are uniformly random and distinct. We then ask her to distinguish between a truly random permutation and the pseudorandom permutation.
(See Section \ref{security} for a precise definition of the security in our setup.)\\
In the present paper, we operate in the setting of the random oracle model (see \cite{randomoracle}). Our contribution is to give a pseudorandom permutation in the bounded retrieval model and prove that it is secure in the weak sense that is discussed above.

\indent Just as in \cite{big-key} we use a big key. We now give a brief sketch of our approach. Note that if one fixes the string of leaked bits, the key is a uniform sample from the preimage of the leaked string. If the length of the leaked string is small, then on average the preimage is very large. What this means is that even when the leakage is known, with a high probability the total entropy of the key is high. This implies that the sum of entropies of each bit in our key is large. This means that many of the bits in the key are ``unpredictable'' in the sense that the probabilty of $1$ is not close to $0$ or $1$. So, if one uses a random oracle to look at various positions of the key and take an XOR, it is likely that the resulting bit is close to an unbiased random bit. This idea of probing the key is similar to the one used in \cite{big-key}. The content of Sections \ref{technical_section} and \ref{Randomnessofbitsection}, which form the heart of this paper, is to show that bits generated by probing the key are close to i.i.d. unbiased random bits. To construct a pseudorandom permutation using these bits, we use a particular card shuffling scheme called the Thorp shuffle, just as in $\cite{proceedings}$. This construction is given in the next section.

\section{Thorp shuffle/maximally unbalanced Feistel network }\label{definition}

One method of turning a pseudorandom function into a pseudorandom permutation is to use a Feistel network (see \cite{lr}). The {\it maximally unbalanced} Feistel network is also known as the Thorp shuffle. Round $r$ of this shuffle can be described as follows. Suppose that the current binary string (i.e., the value of the message after the first $r-1$ rounds of encryption) is $LR$, where $\length(L) = 1$ and $\length(R) = \messagelength-1$. Then round
$r$ transforms the string to $RL'$, where
\[
L' = L \oplus F_\keyval(R, r),
\]
and $F_\keyval$ is a pseudorandom function. Let $X_t(m)$ denote the result of $t$ Thorp shuffles on message $m$. The novel idea in the present paper is to use a pseudorandom function based on a big key $\keyval$.\\
~\\
{\bf The Big Key Pseudorandom Function:} 
Let $\keylength$ be the length of the big key $\keyval$. To compute $F_k(R, r)$, apply the random oracle to $(R, r)$  to obtain $\big(P, \sset \big)$, where $P = (P_1,\ldots, P_\numprobes)$  is $\numprobes$ samples with replacement from $\{1,\ldots, \keylength\}$ and $\Scal$ is a uniform random subset of $\{1,\ldots,n\}$ that is independent of $P$. By analogy with \cite{big-key}, we define the random subkey by $\keyval[P] := (\keyval[P_1], \dots, \keyval[P_\numprobes])$. Finally, define
\[
F_k(R, r) = \oplus_{i \in \sset} \keyval[P_i] \;.
\]
That is, $F_k(R, r)$ is the XOR of a randomly chosen subsequence of the subkey. For a given key $k$ we define our cipher to be $X_\totaltime(\cdot)$ for some fixed positive integer $\totaltime$.\\
\\
{\bf Remark:} We conjecture that it would also work (i.e., we would get a suitable pseudorandom function)
if we took the XOR of the entire subkey;
the current definition is used because it makes the proof simpler. \\
~\\

\section{Security of the Cipher}\label{security}
In this section we introduce a notion of security for pseudorandom permutations under the assumption that there is a leakage of data.
Let $\keyspace = \{0,1\}^\keylength$ denote the set of keys and let $\messagespace = \{0,1\}^\messagelength$ denote the set of messages. We assume that the adversary can leak $\leakagelength$ bits of data and just as in \cite{big-key}, use a function $\leakagefunc : \keyspace\rightarrow\{0,1\}^\leakagelength$ to model this. Henceforth we will refer to this function as the \textit{leakage function}. The adversary has the power to choose this function and this function can depend on calls to the random oracle. For a key $\keyvariable$, we will use $\leakagevariable =\leakagefunc(\keyvariable)$ to denote the output one gets by applying the leakage function to it. We will call this the \textit{leakage}. We allow the adversary to make $\numoraclecalls$ random oracle calls and decide on a leakage function $\leakagefunc$. After the adversary has chosen a leakage function, consider the following two worlds.\\
\textbf{World 1}: In this world, we first choose distinct uniformly random messages $M_1,\ldots, M_\numqueries \in \messagespace$. Then, for a uniformly random key $\keyvariable\in\keyspace$, we set $C_i = X_\totaltime(M_i)$ where $X_t$ is the Thorp shuffle based cipher we defined in Section $\ref{definition}$ and $\totaltime$ is some fixed positive integer. We give the adversary access to the leakage $\leakagevariable$, the input-output pairs $(M_1,C_1), \ldots, (M_\numqueries, C_\numqueries)$ and the random oracle calls that were used by the algorithm
to compute the $X_\totaltime(M_i)'s$.\\
\textbf{World 0}: In this world, again we choose distinct uniformly random messages $M_1,\ldots, M_\numqueries \in \messagespace$. We once again choose a random key $\keyvariable\in \keyspace$ and compute $\leakagevariable$ and all the random oracle calls necessary to evaluate the $X_\totaltime(M_i)'s$, just like world 1. However, instead of setting $C_i's$ to be the outputs of the cipher, we choose a uniformly random
permutation $\pi:\messagespace\to\messagespace$ and set $C_i = \pi(M_i)$. Just as in world 1, the adversary is provided access to the input-output pairs for the $\numqueries$ messages, the leakage $\leakagevariable$ and the random oracle calls.\\
\indent We now place him in these two worlds one at a time without telling him which world he is in. In each of these cases we ask the adversary to guess which world he is in. Let $\mathcal{A}(0)$ and $\mathcal{A}(1)$ denote the answers he gives in world 0 and world 1 respectively. Then, we define the advantage of an adversary as 
\begin{equation}
	\textbf{Adv}(\mathcal{A}) = \prob_1\big(\mathcal{A}(1)=1 \big) - \prob_0\big(\mathcal{A}(0)=1 \big),
\end{equation} 
where $\prob_i$ is the probability measure in world $i$. Define the maximum advantage
\begin{equation}
	\mathbf{MaxAdv}_{\mathbf{\numoraclecalls,\numqueries}} = \max_{\mathcal{A}}\Big( \textbf{Adv}(\mathcal{A})\Big),
\end{equation}
where the maximum is taken over all adversaries satisfying the above mentioned conditions. Note that in the above setup if we allow
the messages to be chosen by the adversary instead of being random, we get the notion of security of a block cipher against chosen plain text attack (CPA) under leakage. Security against CPA is weaker than security against CCA (chosen ciphertext attack). Unfortunately, if a leakage is allowed, it is not possible to design a cipher that is secure in the CPA framework. This is because of the adversary who does the
following: Let $\numqueries=1$. Assume that the message length $\messagelength$ is less than $\leakagelength$. For each key $\keyval$, the adversary includes the ciphertext $X_\totaltime(M_1)$ into the leakage, for a fixed message $M_1$. Then, the adversary answers as follows. If $C_1 = X_\totaltime(M_1)$ then the adversary guesses that he is in world 1. Else, the guess is world 0. In this case, $\prob_1 (\mathcal{A}(1)=1)=1$ and $\prob_0(\mathcal{A}(0)=1)=1/2^\messagelength$. Hence this adversary has a very high advantage. By instead providing the adversary with uniform random plaintext-ciphertext pairs, we get the notion of secuirty against a known plain text attack (KPA) under leakage. The main result of this paper is the following bound on the maximum advantage of such an adversary.
\begin{theorem}    \label{maintheorem}
	The adversary's advantage satisfies
	\[
	\advantage \leq \frac{\numqueries}{\s+1}\bigg(\frac{4\messagelength\numqueries}{2^\messagelength} \bigg) ^\s
	+
	{\numqueries\totaltime \over 2}
	\Bigl[ h^{-1} \Bigl(1 - {\alpha + \numprobes \over \keylength} \Bigr) \Bigr]^{\numprobes/2} + \frac{ \numqueries \numoraclecalls }{2^{\messagelength-1}} + \frac{\numqueries\totaltime}{2^\messagelength} , \,
	\]
	where  $\numoraclecalls$ is the number of random oracle calls, $\alpha = \leakagelength + \messagelength(\numqueries + 1) + \totaltime$, $\s$ is an integer satisfying the equation $\totaltime = \s(\messagelength - 1)$ and $h^{-1}$ is the inverse of the function $h$ restricted to $[1/2, 1]$, where $h$ is defined by $h(p) = p \log_2 \frac{1}{p} + (1-p)\log_2 \frac{1}{p-1}$.
\end{theorem}
Lets try to make sense of this bound. The first two terms have exponents that we can control by choosing the parameters of the cipher. Specifically, we can, with modest assumptions on the number of queries and amount of leakage, make the first term as small as desired by running the cipher for $\totaltime = \mathcal{O}(
\log(\numqueries))$ rounds and make the second term equally small by sampling $\numprobes = \mathcal{O}(\log(\numqueries))$ probes in each round. \\
\indent To make sense of the last two terms, lets consider an adversary which we will call the \textit{naive adversary}. The naive adversary chooses a set $\messagespace'$ of $\lfloor \frac{\leakagelength}{\messagelength} \rfloor$ messages and uses their $\leakagelength$ bits of leakage to leak the ciphertext of each message in $\messagespace'$. Next, when placed in either world 0 or world 1, the naive adversary checks if any of the q random messages provided is from the collection $\messagespace'$. The naive adversary answers "world 1" if the corresponding ciphertext matches the laked ciphertext. Otherwise, they answer "world 0". If none of the $q$ provided messages are from $\messagespace'$, then the naive adversary answers based on the flip of an independent fair coin. Let $\mathbf{Adv_{naive}}$ denote the advantage of the naive adversary. Then,
\[
	\mathbf{Adv_{naive}} = \prob(M_i \in \messagespace' \text{ for some } 1 \leq i \leq \numqueries)(1-2^{-\messagelength}). \,
	\]
 Recall that the distinct messages $M_1,\ldots,M_q$ are sampled uniformly, and that $|\messagespace'| = \lfloor \frac{\leakagelength}{\messagelength} \rfloor$. Let $X = |\{M_1,\ldots,M_q\} \cap \messagespace'|$. Then $X$ is a hypergeometric random variable and
 \[
	\prob(M_i \in \messagespace' \text{ for some } 1 \leq i \leq \numqueries) = \prob(X > 0). \,
	\]
 Using the bound
  \[
	\prob(X > 0) \geq \frac{\e X}{1 + \e X} \,
	\]
for hypergeometric random variables, we get
\[
	\mathbf{Adv_{naive}} \geq \frac{\numqueries \lfloor \frac{\leakagelength}{\messagelength} \rfloor 2^{-\messagelength}}{1 + \numqueries \lfloor \frac{\leakagelength}{\messagelength} \rfloor 2^{-\messagelength}} \cdot (1-2^{-\messagelength}). \,
	\]
 With the modest assumption that $\numqueries \lfloor \frac{\leakagelength}{\messagelength} \rfloor \leq 2^\messagelength$, and the fact that $\messagelength \geq 1$, we can simplify this bound to
 \[
	\mathbf{Adv_{naive}} \geq \frac{\numqueries \lfloor \frac{\leakagelength}{\messagelength} \rfloor}{4 \cdot 2^m}. \,
	\]
 Returning to the bound of the advantage of the optimal adversary, if we assume that $\numoraclecalls,\totaltime \leq \lfloor \frac{\leakagelength}{\messagelength} \rfloor$ and $\numqueries \lfloor \frac{\leakagelength}{\messagelength} \rfloor \leq 2^m$, then
\[
	\advantage \leq \frac{\numqueries}{\s+1}\bigg(\frac{4\messagelength\numqueries}{2^\messagelength} \bigg) ^\s
	+
	{\numqueries\totaltime \over 2}
	\Bigl[ h^{-1} \Bigl(1 - {\alpha + \numprobes \over \keylength} \Bigr) \Bigr]^{\numprobes/2} + 12 \cdot \mathbf{Adv_{naive}} . \,
	\]
 Thus, with realistic assumptions, no adversary can do much better than the naive adversary. To make this precise consider the following example. Let $\keylength = 2^{43}$ bits and $\leakagelength = \keylength/8 = 2^{40}$ bits, i.e. the key has a size of 1 terabyte out of which, $12.5\%$ or about 125 gigabytes can be leaked. Assume that the message length is $\messagelength = 128$ bits. Fix $\numprobes = 500$, $\s=2$ and $\totaltime = \s (2\messagelength - 1) = 510$. Let $\Gamma(\numqueries)$ denote the two leading terms on the RHS of the above inequality, i.e.
 \[
 \Gamma(\numqueries) = 
 \frac{\numqueries}{\s+1}\bigg(\frac{4\messagelength\numqueries}{2^\messagelength} \bigg) ^\s
 +
 {\numqueries\totaltime \over 2}
 \Bigl[ h^{-1} \Bigl(1 - {\alpha + \numprobes \over \keylength} \Bigr) \Bigr]^{\numprobes/2},
 \]
 with values of $\keylength, \leakagelength, \messagelength, \numprobes, \s$ and $\totaltime$ fixed as discussed above. Figure \ref{figureone} shows a plot between $\log_2(\numqueries)$ and $-\log_2(\Gamma(\numqueries))$, for values of $\numqueries$ satisfying $\numqueries\geq 0$, $\numqueries \lfloor \frac{\leakagelength}{\messagelength} \rfloor \leq 2^\messagelength$, $1-(\alpha + \numprobes)/\keylength \geq 0$ and $\Gamma(\numqueries)\leq 1$. From this plot we can see that for the example under consideration, until about $\numqueries = 2^{30}$, any adversary can only have a slightly higher advantage than 12 times the advantage obtained using the naive strategy.
\begin{figure}[!h]
	\centering
	\includegraphics[scale=0.9]{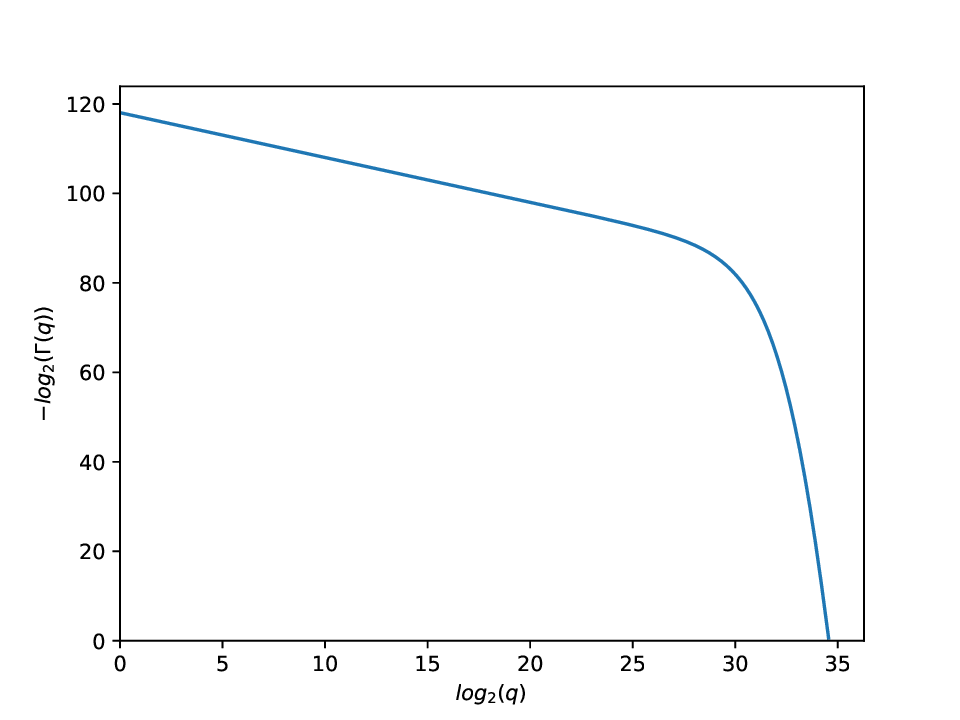}
	\caption{Plot of $-\log_2\big(\Gamma(\numqueries)\big)$ vs $\log_2(\numqueries)$ for a Particular Example.}
	\label{figureone}
\end{figure} \\
 \\
 If a bound avoiding the use of $h^{-1}$ is desired, we can make use of Lemma $\ref{inversebound}$ in Section $\ref{bern}$ which states
 \[
	h^{-1}(z) \leq \half + \half \sqrt{1 - z^{\ln 4}}. \,
	\]
 This gives the following bound on the adversary's advantage:
 \[
	\advantage \leq \frac{\numqueries}{\s+1}\bigg(\frac{4\messagelength\numqueries}{2^\messagelength} \bigg) ^\s
	+
	{\numqueries\totaltime \over 2}
	\left[ \half + \half \sqrt{1 - \Bigl(1 - {\alpha + \numprobes \over \keylength} \Bigr)^{\ln 4}} \right]^{\numprobes/2} + \frac{ \numqueries \numoraclecalls }{2^{\messagelength-1}} + \frac{\numqueries\totaltime}{2^\messagelength} , \,
	\]
\section{Entropy Background and Notation}
Let $X,Y$ be two random variables. Then let $\L(X)$ and $\L(X \given Y)$ denote the law of $X$ and the law of $X$ given $Y$ respectively. For example, let $X$ be a uniform random variable over $\{0,1\}^{\keylength_0}$ and suppose $\leak:\{0,1\}^{\keylength_0} \to \{0,1\}^{\leakagelength_0}$. We write $\L(X \given \leak(X))$ for the {\it random} probability measure $p$ defined by
\[
p(x) :=
\left\{\begin{array}{ll}
	{1 \over | \leak^{-1}( \leak(X))| }           & \mbox{if $\leak(x) = \leak(X)$;} \\[5pt]
	0           & \mbox{otherwise.} \\
\end{array}
\right.
\]
That is, if $\leak(X) = \leakageval$, then $\L(X \given \leak(X))$ is the uniform distribution
over $\{ x \in \{0,1\}^{\keylength_0} : \leak(x) = \leakageval \}$. \\
~\\
Let $\ent$ be the entropy base $2$, that is
\[
\ent(p) := \sum_{x \in \Omega} p(x) \log_2 {1 \over p(x)}  \,.
\]
For a set $S$, define $\ent(S) := \log_2|S|$. That is, $\ent(S)$ is the entropy of the uniform distribution over $S$. 
\begin{lemma}
	\label{divide}
	Let $X$ be a uniform random variable over $A \subset \{0,1\}^{\keylength_0}$ and suppose $\leak:A \to \leakagevals$, where $|\leakagevals| = 2^{\leakagelength_0}$.  Define $S(X) := \leak^{-1}( \leak(X))$. 
	Then 
	\[  
	\e \Bigl( \ent (S(X)) \Bigr) \geq \log_2 |A| - \leakagelength_0.
	\]  
	Furthermore for any $m \in \mathbb{R}$,
	\[
	\P \Bigl( \ent (S(X))  < \log_2 |A| - \leakagelength_0 - m \Bigr)   \leq 2^{-m} \,.
	\]
\end{lemma}    
\begin{proof}
	For $\leakageval \in \leakagevals$, let $S_\leakageval := \{ x: \leak(x) = \leakageval \}$. Note that if $X \in S_\leakageval$ then $S(X) = S_\leakageval$. It follows that
	\begin{eqnarray}
		\e \Bigl( \ent (S(X)) \Bigr)
		&=&  \sum_{\leakageval \in \leakagevals}
		{|S_\leakageval| \over |A|} \log_2 |S_\leakageval|  \nonumber  \\
		&=&  \log_2 |A| + \sum_{\leakageval \in \leakagevals}
		{|S_\leakageval| \over |A|} \log_2 {|S_\leakageval| \over |A|}  \nonumber  \\
		&=&  \log_2 |A| +  \,|\leakagevals| \Bigl[ {1 \over |\leakagevals|}      \sum_{\leakageval \in \leakagevals}
		{|S_\leakageval| \over |A|} \log_2 {|S_\leakageval| \over |A|}    \Bigr] \,. \label{lucy}
	\end{eqnarray}
	The average (over $\leakageval$) of the quantity ${|S_\leakageval| \over |A|}$ is ${1 \over |\leakagevals|}$. Therefore, since the function $x \log x$ is convex, Jensen's inequality implies that the quantity (\ref{lucy}) is at least
	\[   
	\log_2 |A| +  |\leakagevals| \left( {1 \over |\leakagevals|}  \log_2\left( {1 \over |\leakagevals|} \right)\right) 
	= \log_2 |A| - \leakagelength_0 \,.
	\]
	For the second part of the lemma, note that
	\begin{eqnarray*}
		\P\Bigl( \ent( S(X) ) < \log_2 |A| - \leakagelength_0 - m \Bigr) &=& \P( |S(X)| <  |A| \cdot 2^{-\leakagelength_0 -m} ) \\
		&=& \sum {|S_\leakageval| \over |A|}, \\
	\end{eqnarray*}
	where the sum is over $\leakageval$ such that $|S_\leakageval| < |A| \cdot 2^{-\leakagelength_0-m}$. Since each term in the sum is at most $2^{-\leakagelength_0-m}$ and there are at most $2^{\leakagelength_0}$ terms, the sum is at most $2^{-m}$.
\end{proof}

\section{Entropy and Bernoulli Random Variables}  \label{bern}
Let $h(p)$ denote the entropy of a Bernoulli($p$) random variable. That is,
define $h:[0,1] \to [0,1]$ by
\[
h(p) = p \log_2{1 \over p} + (1-p) \log_2{1 \over 1-p} \,.
\]
The restriction of $h$ to $[\half, 1]$ is a strictly decreasing and onto function and hence has an inverse
$h^{-1}:[0,1] \to [\half, 1]$. Since $h$ is concave and decreasing on $[\half, 1]$, the function $h^{-1}$ is concave.
Furthermore, note that for any $p \in [0,1]$, we have $h(p) = h(1-p)$ and hence
\begin{equation}
	\label{max}
	h^{-1}(h(p)) = \max(p, 1-p) \,.
\end{equation}
Theorem 1.2 of \cite{entropy} gives the
following bound:
\begin{equation}
	\label{ebound}
	h(p) \leq (4pq)^{1/\ln4},
\end{equation}
where $q = 1-p$. This implies the following lemma. 
\begin{lemma} \label{inversebound}
	For any $p \in [0,1]$, we have
	\[
	p \leq \half \left( 1 + \sqrt{1 - h(p)^{\ln 4}} \right)
	\]
\end{lemma}
\begin{proof}
	Let $\Delta = \max(p, 1-p) - \half$. Then
	\begin{eqnarray*}
		pq &=& \left(\Half + \Delta\right)\left(\Half - \Delta\right) \\
		&=& \Quarter - \Delta^2,
	\end{eqnarray*}
	and hence
	\begin{equation}
		\label{delta}
		4pq = 1 - 4 \Delta^2.
	\end{equation}
	Equation (\ref{ebound}) implies that
	\[
	4pq \geq h(p)^{\ln 4}.
	\]
	Combining this with (\ref{delta}) gives
	\[
	\Delta^2 \leq {1 \over 4} (1 - h(p)^{\ln 4}),
	\]
	and hence
	\[
	\Delta \leq {1 \over 2} \sqrt{1 - h(p)^{\ln 4}} \,.
	\]
	It follows that
	\begin{eqnarray*}
		p &\leq& {1 \over 2} + \Delta \\
		&\leq& {1 \over 2} + {1 \over 2} \sqrt{1 - h(p)^{\ln 4}} \\
		&=& {1 \over 2} \left(1 + \sqrt{1 - h(p)^{\ln 4}} \right).
	\end{eqnarray*}
\end{proof} 

Recall that $h(p)$ is the entropy of a Bernoulli($p$) random variable and if $S \subset \{0,1\}^\keylength$ then $\ent(S) := \log_2 |S|$. We shall need the following entropy decomposition lemma.
\begin{lemma}\label{decomposition}
	Suppose that $S \subset \{0,1\}^\keylength$ and suppose that $\keyvariable$ is uniformly distributed over $S$. 
	Then
	\begin{equation*}
		\sum_{i=1}^\keylength h \left( \prob(\keyvariable[i] = 1) \right) \geq \ent(S) .
	\end{equation*}
\end{lemma}
\begin{proof}
	Note  that $\ent(S)$ is the entropy of $\keyvariable$.
	Applying the chain rule for entropy on $\keyvariable=(\keyvariable[1],\ldots,\keyvariable[\keylength])$ gives
	\begin{equation*}
		\sum_{i=1}^{\keylength} \ent((\keyvariable[i] \given \keyvariable[i-1],\keyvariable[i-1],\ldots,\keyvariable[1]\big)) = \ent(S) \, .
	\end{equation*}
	It is well known that for any two discrete random variables $Z,Z'$ on a common probability space, $\mathbf{H}(Z|Z')\leq \mathbf{H}(Z)$. So the above inequality gives
	\begin{equation*}
		\sum_{i=1}^\keylength \ent((\keyvariable[i])) \geq \sum_{i=1}^{\keylength} \mathbf{H}\big(\keyvariable[i]\ \big|\ \keyvariable[i-1],\ldots,\keyvariable[1]\big) = \ent(S) \, .
	\end{equation*}
	Finally, note that $\ent( \keyvariable[i])  = h \left( \prob(\keyvariable[i] = 1) \right)$.
\end{proof}

\section{Main Technical Results}\label{technical_section}
\begin{lemma}\label{subprobelemma}
	Let $Y=(Y_1,Y_2,\ldots,Y_n)\in \{0,1 \}^n$ be a random n-bit string. For $S\subseteq \{1,\ldots,n\}$, set $f_S(Y):=(-1)^{\oplus_{i\in S} Y_i}$, with the convention that $f_{\emptyset}\equiv 1$. Also let 
	\begin{eqnarray}
		\Ecal(S) &=& \E [f_S(Y)]     \nonumber  \\
		&=&  2 \Bigl( \half - \prob ( \oplus_{i\in S} Y_i = 1) \Bigr) \; .    \label{tvrel}
	\end{eqnarray}
	Then for a uniformly chosen random subset $\Scal\subseteq\{1,\ldots, n\}$, we have
	\begin{equation}
		\E\big[\Ecal(\Scal)^2] = \sum_{y\in \{0,1\}^n} \probb{Y=y}^2.
	\end{equation}
\end{lemma}
Lemma
\ref{subprobelemma} is a well-known consequence of Parseval's theorem (see page 24 of \cite{odonnell}). For completeness, we give a proof here: \\
~\\
\begin{proof}
	Let $\Omega = \{0,1\}^n$. Note that $\mathbbm{R}^\Omega$, the space of real valued functions on $\Omega$, forms a vector space of dimension $|\Omega|$ over $\mathbbm{R}$. Define the following inner product on $\mathbbm{R}^\Omega$.
	\begin{equation*}
		\begin{split}
			\langle f,g\rangle &= \frac{1}{2^n}\sum_{x\in\Omega}f(x)g(x) \quad\text{for } f,g\in \mathbbm{R}^\Omega\\
			&=\E[f(Z)g(Z)],
		\end{split}
	\end{equation*}
	where $Z=(Z_1,\ldots,Z_n)$ and $Z_1,Z_2,\ldots,Z_n$ are i.i.d Bernoulli(1/2) random variables. Observe that when $S\neq S'$,
	\begin{equation*}
		\langle f_S, f_{S'}\rangle =\E [f_S(Z)f_{S'}(Z)]=\E\Bigg[\prod_{i\in S\cap S'}(-1)^{2Z_i}\prod_{j\in S\triangle S'}(-1)^{Z_j}\Bigg]=\prod_{i\in S\cap S'}\E\Big[(-1)^{2Z_i}\Big]\prod_{j\in S\triangle S'}\E\Big[(-1)^{Z_j}\Big]=0,
	\end{equation*}
	since $\E [(-1)^{Z_i}]=0$ and $S\triangle S'$ is non-empty when $S\neq S'$. Also observe that
	\begin{equation*}
		\langle f_S, f_S\rangle = \E\Big[(-1)^{2(\oplus_{i\in S} Z_i)} \Big] = \E[1] = 1.
	\end{equation*}
	Therefore, $\{f_S\}_{S\in 2^{[n]}}$ forms an orthonormal basis for $2^\Omega$. Next, let $U(y)=1/2^n$
	and $P(y)=\probb{Y=y}$ for $y\in\Omega$. Then, $P,U,P/U\in\mathbb{R}^\Omega$. Now note that
	\begin{equation*}
		\langle P/U, f_S\rangle = \frac{1}{2^n}\sum_{x\in \Omega} 2^n P(x) f_S(x) = \sum_{x\in\Omega} P(x)f_S(x) = \E[f_S(Y)]=\Ecal(S) \, .
	\end{equation*}
	It follows that
	\begin{equation*}
		\frac{1}{2^n} \langle P/U, f_S\rangle^2 = \frac{1}{2^n}\Ecal(S)^2.
	\end{equation*}
	Summing the above equation over all subsets $S\subseteq \{1,\ldots,n\}$ and using the fact that $f_S's$ form an orthonormal basis, we get
	\begin{equation*}
		\frac{1}{2^n}\langle P/U, P/U\rangle =\sum_{S\subseteq\{1,\ldots, n\}}\frac{1}{2^n}\langle P/U, f_S\rangle^2 =
		\sum_{S\subseteq\{1,\ldots, n\} }\frac{1}{2^n}\Ecal(S)^2 = \E\big[\Ecal(\Scal)^2\big].
	\end{equation*}
	The left hand side of the above equation simplifies to $\sum_{y\in\Omega} P(y)^2$ and hence the proof is complete.
\end{proof}
Note that $\Ecal(S)$ is a measure of the bias in the parity of the bits of $Y$ whose positions are in $S$. More precisely, recall that for  probability distributions $\mu$ and $\nu$ on a finite set $\Omega$, the total variation distance 
\begin{equation}
	\lVert \mu - \nu\rVert_{TV} := \frac{1}{2}\sum_{x\in\Omega} |\mu(x)-\nu(x)|.
\end{equation}
For a $\{0,1\}$-valued random variable $W$, the total variation distance
\[
\lVert W - \coinflip  \rVert_{TV} = 
| \prob(W = 1) - \half | \; .
\]
Equation (\ref{tvrel}) implies that
\[
\half - \prob ( \oplus_{i\in S} Y_i = 1) = \half \Ecal(S),
\]
and hence
\begin{equation}
	\lVert  \oplus_{i\in S} Y_i - \coinflip   \rVert_{TV} = \half | \Ecal(S) |.   \label{tvr}
\end{equation}
For $S\subseteq \{1,\ldots,n\}$, define $\bias(S) := \lVert  \oplus_{i\in S} Y_i - \coinflip \rVert_{TV}$. 
Then
\begin{eqnarray*}
	\e( \bias(\Scal)) &=& \half \e | \Ecal(\Scal) | \\
	&\leq& \half \sqrt { \e( \Ecal( \Scal )^2 )  } \\
	&=&
	\half \Biggl[ \sum_{y\in \{0,1\}^n} \probb{Y=y}^2 \Biggr]^{1/2},
\end{eqnarray*}
where the first line follows from equation (\ref{tvr}), the second line follows from Jensen's inequality and the third line follows from Lemma \ref{subprobelemma}. This leads to the following:
\begin{corollary}\label{tvltwo}
	Let $\keyvariable$ be a random string in $\{0,1\}^\keylength$. Let $(p_1,\ldots, p_\numprobes)$ be a choice of probes. Let  $c'$ be a Bernoulli($1/2$) random variable, and for $S \subset \{1, \dots, \numprobes\}$, define
	\[
	c(S) := \oplus_{i\in S} \keyvariable[p_i];  \hspace{.5 in}
	d(S) := \lVert c(S) - c' \rVert_{TV}.
	\]
	If $\Scal$ is a uniform random subset of $\{1, 2, \dots, \numprobes\}$ then
	\begin{equation*}
		\E (d (\Scal))   \leq \frac{1}{2}\E\Bigg[\sqrt{\sum_{y\in\{0,1\}^\numprobes} \probp{\keyvariable[p_1, \dots, p_\numprobes] =y\ }^2 \ }\Bigg].
	\end{equation*}
\end{corollary}
\noindent
This shows that the expectation (taken over the subprobes) of the  distance between the random bit $c$ and a Bernoulli($1/2$) random variable can be bounded in terms of the  $l^2$-norm of the distribution of $\keyvariable[p_1, \dots, p_\numprobes]$.

\section{Main Lemma} \label{Randomnessofbitsection}
Suppose $\keyvariable$ is a uniform random element of $\{0,1\}^\keylength$ and  suppose $\leak: \{0, 1\}^\keylength \to \leakagevals$. For $\leakageval \in \leakagevals$, let $S_\leakageval := \leak^{-1}(\leakageval)$.  Define the probability measure $\probl$ by
\[
\probl(\, \cdot ) := \prob(\, \cdot \given \leak(\keyvariable) = \leakageval) \,,
\]
and write $\el$ for the expectation operator with respect to $\probl$. Note that under $\probl$, the distribution of $\keyvariable$ is uniform over $S_\leakageval$. 
For an integer $r$ with $1 \leq r \leq \numprobes$ and probes $p_1, \dots, p_r$
define
\[
\g_\leakageval(p_1, \dots, p_r) := \sum_{x \in \{0,1\}^r} \probl( \keyvariable[p_1, \dots, p_r] = x)^2 \,.
\]
\begin{lemma}    \label{mainlemma}
	Suppose that probes $P_1, P_2, \dots, P_\numprobes$ are chosen independently and 
	uniformly at random from $\{1, 2, \dots, \keylength\}$. If $\ent(S_\leakageval) \geq \keylength - \alpha$, then
	\[
	\el ( \g_\leakageval(P_1, \dots, P_\numprobes)) \leq  \Bigl[ h^{-1} \Bigl(1 - {\alpha + \numprobes \over \keylength} \Bigr) \Bigr]^\numprobes \,.
	\]
\end{lemma}
\begin{proof}
Fix $\leakageval$ with $\ent(S_\leakageval) \geq \keylength - \alpha$. For $x \in \{0,1\}^r$, and a choice of probes $p_1, \dots, p_{r+1}$, define
\begin{eqnarray*}
	\lam_\leakageval(x, p_1, \dots, p_{r+1}) &:=& \probl( \keyvariable[p_{r+1}] = 1 \given \keyvariable[p_1 , \dots, p_r] = x )  \\
	&=& \prob( \keyvariable[p_{r+1}] = 1 \given \keyvariable[p_1 , \dots, p_r] = x, \leak(\keyvariable) = \leakageval ) \,.  \\
\end{eqnarray*}
Define $\mur(x) : \probl( \keyvariable[p_1, \dots, p_r] = x)$. Note that conditional on $\leak(\keyvariable) = \leakageval$ and $\keyvariable[p_1, \dots, p_r] = x$, the distribution of $\keyvariable$ is uniform over $S_\leakageval \cap \{A \in \{0,1\}^\keylength : A[p_1, \dots, p_r] = x \}$. Furthermore,
\[
| S_\leakageval \cap \{A \in \{0,1\}^\keylength : A[p_1, \dots, p_r] = x \}| = |S_\leakageval| \cdot \mur(x) \,.
\]
Hence Lemma \ref{decomposition} implies that
\begin{equation}
	\label{ave}
	{1 \over \keylength} \sum_{j=1}^\keylength h( \lam(x, p_1, \dots, p_r, j)) \geq \frac{1}{\keylength} \log_2 \left( |S_\leakageval| \cdot \mur(x) \right) \,.
\end{equation}
For any $p_1, \dots, p_{r+1}$, we have 
\begin{eqnarray}   
	& &  \g_\leakageval( p_1, \dots, p_{r+1})  \\
	&=&  \sum_{y \in \{0,1\}^{r+1}}   \probl( \keyvariable[p_1, \dots, p_{r+1}] = y)^2 \\
	&=&
	\sum_{x \in \{0,1\}^{r}}   \probl( \keyvariable[p_1, \dots, p_{r}] = x)^2 
	\Bigl[ \lam_\leakageval(x, p_1, \dots, p_{r+1})^2 + (1 - \lam_\leakageval(x, p_1, \dots, p_{r+1}))^2 \Bigr ]    \\
	&=&
	\sum_{x \in \{0,1\}^{r}}   \mur(x)^2
	\Bigl[ \lam_\leakageval(x, p_1, \dots, p_{r+1})^2 + (1 - \lam_\leakageval(x, p_1, \dots, p_{r+1}))^2 \Bigr ] \,.  \label{square}
\end{eqnarray}
Note that for any $p \in [0,1]$ we have $p^2 + (1 - p)^2 \leq \max(p, 1-p)$. Hence, the quantity in square brackets in equation (\ref{square})  is at most
\[
h^{-1}( h( \lam_\leakageval(x, p_1, \dots, p_{r+1})))
\]
by equation (\ref{max}). Thus
\begin{eqnarray}   
	\g_\leakageval( p_1, \dots, p_{r+1})
	&\leq&
	\sum_{x \in \{0,1\}^{r}}   \mur(x)^2
	h^{-1}( h( \lam_\leakageval(x, p_1, \dots, p_{r+1}))) \,.  \label{square}
\end{eqnarray}
Recall that the probe $P_{r+1}$ is chosen uniformly at random from $\{1, 2, \dots, \keylength\}$. It follows that 
\begin{eqnarray}   
	\el(\g_\leakageval( p_1, \dots, p_{r}, P_{r+1}))
	&\leq&
	{1 \over \keylength} \sum_{j = 1}^\keylength \sum_{x \in \{0,1\}^{r}}   \mur(x)^2
	h^{-1}( h( \lam_\leakageval(x, p_1, \dots, p_r, j))) \\
	&=&
	\sum_{x \in \{0,1\}^{r}}   \mur(x)^2 \Bigl[
	{1 \over \keylength} \sum_{j = 1}^\keylength
	h^{-1}( h( \lam_\leakageval(x, p_1, \dots, p_r, j))) \Bigr] \, .    \label{seventeen}
\end{eqnarray}
Recall that in Section \ref{bern} we showed that $h^{-1}$ is concave. Thus, Jensen's inequality implies that the quantity (\ref{seventeen}) is at most
\begin{eqnarray}
	& & \sum_{x \in \{0,1\}^{r}}   \mur(x)^2 h^{-1}\Bigl(
	{1 \over \keylength} \sum_{j = 1}^\keylength
	h( \lam_\leakageval(x, p_1, \dots, p_r, j)) \Bigr) \\ 
	&\leq& \sum_{x \in \{0,1\}^{r}}   \mur(x)^2
	h^{-1}\Bigl(
	{1 \over \keylength} \log_2( |S_\leakageval| \cdot \mur(x)) \Bigr), \\     \label{eighteen}
\end{eqnarray}
where the inequality follows from (\ref{ave}) and the fact that $h^{-1}$ is decreasing. Recall that the Harris-FKG inequality (see Section 2.2 of \cite{grimmett}) implies that if $X$ is a random variable and
$f$ (respectively, $g$) is an increasing (respectively, decreasing) function, then
\[
\e( f(X) g(X) ) \leq \e(f(X)) \e(g(X)) \,.
\]
Now, consider the probability measure that assigns mass $\mur(x)$ to each $x \in \{0,1\}^r$, and let $X: \{0,1\}^r \to \reals$ be the random variable defined by $X(x) = \mur(x)$.  Let $f$ be the identity function on $\reals$ and define $g: \reals \to \reals$ by $g(u) = h^{-1}\Bigl( {1 \over \keylength} \log_2( |S_\leakageval| \cdot u) \Bigr)$. Then $f$ is increasing and since $h^{-1}$ is decreasing, $g$ is decreasing. Thus the Harris-FKG inequality implies that the quantity (\ref{eighteen}) is at most
\begin{eqnarray}
	& & \Bigl(  \sum_{x \in \{0,1\}^{r}}   \mur(x)^2 \Bigr)
	\Bigl(  \sum_{x \in \{0,1\}^{r}}   \mur(x)
	h^{-1}\Bigl(
	{1 \over \keylength} \log_2( |S_\leakageval| \cdot \mur(x)) \Bigr) \\
	&\leq&
	\Bigl(  \sum_{x \in \{0,1\}^{r}}   \mur(x)^2 \Bigr)
	h^{-1}\Bigl( {1 \over \keylength}
	\sum_{x \in \{0,1\}^{r}}   \mur(x)
	\log_2( |S_\leakageval| \cdot \mur(x)) \Bigr) \, .     \label{al}
\end{eqnarray}
Applying the first part of Lemma \ref{divide} with $\leak(\keyvariable) = \keyvariable[p_1, p_2, \dots, p_r]$
gives
\begin{eqnarray*}
	\sum_{x \in \{0,1\}^{r}}   \mur(x)
	\log_2( |S_\leakageval| \cdot \mur(x)) &\geq& 
	\ent(S_\leakageval) - r  \\
	&\geq& \keylength - \alpha - \numprobes,
\end{eqnarray*}
where the second inequality follows from the fact that $\ent(S_\leakageval) \geq \keylength - \alpha$ and $r \geq \numprobes$.  It follows that the quantity (\ref{al}) is at most
\begin{eqnarray}
	& &
	\Bigl(  \sum_{x \in \{0,1\}^{r}}   \mur(x)^2 \Bigr)
	h^{-1}\Bigl( {1 \over \keylength}(\keylength - \alpha - \numprobes) \Bigr)  \\
	&=&
	\Bigl(  \sum_{x \in \{0,1\}^{r}}   \mur(x)^2 \Bigr)
	h^{-1}\Bigl( 1 - {\alpha + \numprobes \over \keylength} \Bigr)  \\
	&=&
	\g_l(p_1, p_2, \dots, p_r)
	h^{-1}\Bigl( 1 - {\alpha + \numprobes \over \keylength} \Bigr) \,.
\end{eqnarray}
We have shown that for any choice of $p_1, p_2, \dots, p_r$ we have
\[
\el(\g_\leakageval( p_1, \dots, p_{r}, P_{r+1})) \leq       \g_\leakageval(p_1, p_2, \dots, p_r)
h^{-1}\Bigl( 1 - {\alpha + \numprobes \over \keylength} \Bigr) \,.
\]
It follows that
\[
\el(\g_\leakageval( P_1, \dots, P_{r+1})) \leq  \el(\g_\leakageval(P_1, P_2, \dots, P_r))
h^{-1}\Bigl( 1 - {\alpha + \numprobes \over \keylength} \Bigr) \,.
\]
Since this is true for all $r$ with $1 \leq r \leq \numprobes$, we have
\[
\el(\g_\leakageval( P_1, \dots, P_\numprobes)) \leq  \el(\g_\leakageval(P_1))
\Bigl[ h^{-1}\Bigl( 1 - {\alpha + \numprobes \over \keylength} \Bigr)  \Bigr]^{\numprobes-1} \,.
\]
Finally, an argument similar to above (eliminating all sums over $\{0,1\}^r$ and replacing $\mur(x)$ by $1$) shows that
\[
\el( \g_\leakageval( P_1)) \leq h^{-1}\Bigl( 1 - {\alpha + \numprobes \over \keylength} \Bigr) \,,
\]
and the lemma follows. 
\end{proof}

\section{Proof of Main Theorem}

In this section we will prove Theorem \ref{maintheorem}, which bounds the advantage of a KPA adversary with leakage against the cipher described in section \ref{definition}. Recall that the bound in question is
    \[
	\advantage \leq \frac{\numqueries}{\s+1}\bigg(\frac{4\messagelength\numqueries}{2^\messagelength} \bigg) ^\s
	+
	{\numqueries\totaltime \over 2}
	\Bigl[ h^{-1} \Bigl(1 - {\alpha + \numprobes \over \keylength} \Bigr) \Bigr]^{\numprobes/2} + \frac{ \numqueries \numoraclecalls }{2^{\messagelength-1}} + \frac{\numqueries\totaltime}{2^\messagelength} , \,
	\]
	First, we prove the bound assuming that the adversary makes no random oracle calls.
	Let $( \M_i, \C_i)_{i=1}^\numqueries$ be the uniform random sequence of input/output pairs given to the adversary. \\
	The adversary's advantage satisfies
	\[
	\mathbf{MaxAdv}_{\mathbf{0,\numqueries}} \leq  \lVert (\M_i, C_i)_{i=1}^\numqueries - (\M^u_i, C^u_i)_{i=1}^\numqueries \rVert_{TV},
	\]
	where $(\M^u_i, C^u_i)_{i=1}^\numqueries$ are $\numqueries$ uniform random queries from a uniform random permuation. Let $(\M^\th_i, C^\th_i)_{i=1}^\numqueries$ be $\numqueries$ uniform random queries from $\totaltime$ rounds of an {\it idealized} Thorp shuffle that uses a uniform random round function $F$ instead of a pseudorandom function.
	
	In $\cite{proceedings}$, Morris, Rogaway and Stegers prove the following.
	\begin{theorem}\cite{proceedings}\label{morristhorp}
		Let $\totaltime=\s(2\messagelength-1)$ for some whole number $s$, where $2^\messagelength=|\messagespace|$. Then
		\begin{equation*}
			\lVert (\M^\th_i, C^\th_i)_{i=1}^\numqueries - (\M^u_i, C^u_i)_{i=1}^\numqueries \rVert_{TV}
			\leq \frac{\numqueries}{\s+1}\bigg(\frac{4\messagelength\numqueries}{2^\messagelength} \bigg) ^\s.
		\end{equation*}
	\end{theorem}
	Combining this result with a bound on
	\begin{equation}
		\lVert (\M_i, C_i)_{i=1}^\numqueries - (\M^\th_i, C^\th_i)_{i=1}^\numqueries \rVert_{TV}    \label{compthorp}
	\end{equation}
	will give the claimed bound on the adversary's advantage.
	To bound (\ref{compthorp}) we use a hybrid argument. \\
	~\\
	For $0 \leq l \leq q$, let $C^l_i$ be the result of $\totaltime$ Thorp shuffles starting from $M_i$, where the round function
	used to determine the ``random bits'' of the shuffle is defined as follows:
	\begin{enumerate}
		\item If $i \leq l$ then we use the round function $F_\keyval$.
		\item If $l + 1 \leq i \leq \numqueries$, then at any step not already determined by the round functions used to evaluate
		the first $l$ queries, we use a uniform random function $F$.
	\end{enumerate}
	Define $Q_l = ( \M_i, \C^l_i)_{i=1}^\numqueries$. Thus, the first $l$ queries of $Q_l$ correspond to the Thorp shuffle using the pseudorandom round function $F_k$ and the final $q - l$ queries correspond to the Thorp shuffle using a uniform random round function (except at steps that are already ``forced'' by the trajectories of the first $l$ queries.)
	Note that $Q_0$ corresponds to the Thorp shuffle with a uniform random round function and $Q_\numqueries$ corresponds to the Thorp shuffle with the ``big key'' pseudorandom round function $F_k$. Thus, the triangle inequality gives us
	\begin{eqnarray*}
		\lVert (\M_i, C_i)_{i=1}^\numqueries - (\M^\th_i, C^\th_i)_{i=1}^\numqueries \rVert_{TV}
		&\leq& \sum_{k=0}^{\numqueries-1}
		\lVert Q_{k+1} - Q_k \rVert_{TV}
	\end{eqnarray*}
 To bound the terms of this sum, we prove the following lemma:
\begin{lemma}
	\label{qlemma}
	For all $s$ we have
	\[
	\lVert Q_{s+1} - Q_s \rVert_{TV} 
	\leq {\totaltime \over 2}
	\Bigl[ h^{-1} \Bigl(1 - {\alpha + \numprobes \over \keylength} \Bigr) \Bigr]^{\numprobes/2} \, + \totaltime \cdot 2^{-\messagelength}.
	\]
\end{lemma}
\begin{proof}
	It is sufficient to bound 
	\[
	\lVert Z_s - Z'_s \rVert_{TV},
	\]
	where
	\[
	Z_s = (Q_{s+1}, \T_{s+1}); \hspace{1 in}
	Z'_s = (Q_s, \T_{s+1}) \; .
	\]
	\\
	~\\
	and $\T_{s+1} = (X_0(M_{s+1}), \dots, X_\totaltime(\M_{s+1}))$ is the trajectory of message $\M_{s+1}$. \\
	~\\
	Again, we use a hybrid argument. For $i$ with $1 \leq i \leq \totaltime$, let $P_{s,i}$ be the algorithm defined as follows: \\
	~\\
	{\bf Algorithm $P_{s,i}$:}  For the first $s$ queries and for the first $i$ rounds of query $s+1$, we use the pseudorandom function $F_k$. For rounds $i + 1, \dots, \totaltime$ of query $s+1$ and for queries $s + 2, \dots, \numqueries$, any random bit (that was not already determined by the previous queries) will be defined using a uniform random function $F$. \\
	~\\
	Let $W_{s,i}$ be the value of $\Bigl( ( \M_i, \C_i)_{i=1}^\numqueries, \T_{s+1} \Bigr)$ when algorithm $P_{s,i}$ is followed. Note that $W_{s, \totaltime}$ has the distribution of $Z_s$ and $W_{s, 0}$ has the distribution of $Z'_s$. Therefore, by another application of the triangle inequality, we have 
    \[
	\lVert Z_s - Z_s' \rVert_{TV} \leq \sum\limits_{i=0}^{\totaltime-1}\lVert W_{s, i+1} - W_{s, i} \rVert_{TV}
    \]
	The only difference between $P_{s, i+1}$ and 
	$P_{s, i}$ is the bit used in round $i+1$ of query $s+1$. If this bit was determined by the previous queries, then
	it has the same value in both $P_{s, i+1}$ and 
	$P_{s, i}$. Otherwise, it is a $\coinflip$ random variable in 
	$P_{s, i}$ and it uses the ``big key'' pseudorandom function $F_\keyvariable$ in $P_{s, i + 1}$. It is enough to show that the claimed bound on the total variation distance holds even if we condition on the input messages $M_1, \dots, M_\numqueries$. So let $m_1, \dots, m_\numqueries$ be arbitrary input messages. Let $\leak$ be a function on $\{0,1\}^\keylength$ such that $\leak(\keyvariable)$ encodes
	\begin{enumerate}
		\item $\leakagefunc(\keyvariable)$
		\item the values of $C_1, \dots, C_s$ and $X_1(m_{s+1}), \dots, X_i(m_{s+1})$ when algorithm
		$P_{s, i+1}$ is used with key $K$ and input messages $m_1, \dots, m_s$.
	\end{enumerate}
	Let $\leakagevariable = \leak(\keyvariable)$. Note that there are at most
	\[
	2^{\leakagelength} \cdot \left( 2^\messagelength \right)^s \cdot 2^i = 2^{\leakagelength + \messagelength s + i} \leq 2^{\leakagelength + \messagelength \numqueries + \totaltime} 
	\]
	possible values of $\leakagevariable$.
	Define $S_\leakageval := \leak^{-1}(\leakageval)$.  We can use Lemma \ref{divide} to get a bound on the size of $S_\leakagevariable$ that holds with high probability. More precisely, Lemma \ref{divide}
	\[
	\prob\Big( \ent(|S_\leakagevariable|) \leq \keylength - \leakagelength - \messagelength(\numqueries+1) - \totaltime\Big) \leq 2^{-\messagelength} \,.
	\]
	On the event that $\ent(|S_\leakagevariable|) \leq \keylength - \leakagelength - \messagelength(\numqueries+1) - \totaltime$, we can use Lemma \ref{mainlemma} to bound the total variation distance. Using Lemma \ref{mainlemma} with $\alpha = \leakagelength + \messagelength(\numqueries+1) + \totaltime$ and combining this with Corollary \ref{tvltwo} shows that if $B$ is the random bit generated by Algorithm $P_{s, i+1}$ then 
	\[
	\e\Bigl( \lVert B - \coinflip \rVert_{TV} \Bigr) \leq
	{1 \over 2}
	\Bigl[ h^{-1} \Bigl(1 - {\alpha + \numprobes \over \keylength} \Bigr) \Bigr]^{\numprobes/2} + 2^{-\messagelength} \,.
	\]
    Since this one random bit is only nondeterministic difference between $W_{s,i+1}$ and $W_{s,i}$, we have
	\[
	\lVert W_{s, i+1} - W_{s, i} \rVert_{TV}
	\leq {1 \over 2}
	\Bigl[ h^{-1} \Bigl(1 - {\alpha + \numprobes \over \keylength} \Bigr) \Bigr]^{\numprobes/2} + 2^{-\messagelength} \,.
	\]
    This quantity is independent of $i$, so
    \[
	\lVert Z_s - Z_s' \rVert_{TV} \leq {\totaltime \over 2}
	\Bigl[ h^{-1} \Bigl(1 - {\alpha + \numprobes \over \keylength} \Bigr) \Bigr]^{\numprobes/2} + \totaltime \cdot 2^{-\messagelength} \,.
    \]
\end{proof}
Now we use Lemma \ref{qlemma} to bound the sum,
    \begin{eqnarray*}
		\lVert (\M_i, C_i)_{i=1}^\numqueries - (\M^\th_i, C^\th_i)_{i=1}^\numqueries \rVert_{TV}
		&\leq& \sum_{k=0}^{\numqueries-1}
		\lVert Q_{k+1} - Q_k \rVert_{TV}  \\
		&\leq&
		{\numqueries\totaltime \over 2}
		\Bigl[ h^{-1} \Bigl(1 - {\alpha + \numprobes \over \keylength} \Bigr) \Bigr]^{\numprobes/2} + \numqueries\totaltime\cdot 2^{-\messagelength} \,.
	\end{eqnarray*}
	Combining this with Theorem \ref{morristhorp} and another application of the triangle inequality gives
	\begin{equation}
		\label{fff}
		\lVert (\M_i, C_i)_{i=1}^\numqueries - (\M^u_i, C^u_i)_{i=1}^\numqueries \rVert_{TV} \leq
		{\numqueries\totaltime \over 2}
		\Bigl[ h^{-1} \Bigl(1 - {\alpha + \numprobes \over \keylength} \Bigr) \Bigr]^{\numprobes/2} +
		\numqueries\totaltime \cdot 2^{-\messagelength} +
		\frac{\numqueries}{\s+1}\bigg(\frac{4\messagelength\numqueries}{2^\messagelength} \bigg) ^\s \,.
	\end{equation}
	\\
	Finally, we consider the effect of random oracle calls made by the adversary before calculation of $\leakagefunc$. Let $\calls$ be the set of random oracle calls made by the adversary. Note that
	\[
	\calls = \cup_{i = 1}^T \calls_i,
	\]
	where $\calls_i$ is the set of random oracle calls where the input is $(R,i)$ for some $R$. Let $E$ be the event that at least one of the random oracle calls used to evaluate the $\M_i$ is in $\calls$.  Note that for a uniform random message, the value of $R$ (the rightmost $\messagelength-1$ bits) after any number of Thorp shuffles is uniform over $\{0,1\}^{\messagelength-1}$. Hence, the probability that the random oracle call used in stage $r$ of the shuffle is in $\calls_i$ is $\displaystyle \frac{|\calls_i|}{2^{\messagelength-1}}$. 
	Hence, taking a union bound over queries and time steps gives
	\begin{eqnarray*}
		\prob(E) &\leq& \numqueries \sum_{i=1}^\totaltime \frac{|\calls_i|}{2^{\messagelength-1}}   \\
		&=& \frac{ \numqueries |\calls|}{2^{\messagelength-1}}   \\
		&=& \frac{ \numqueries \numoraclecalls }{2^{\messagelength-1}}.
	\end{eqnarray*}
	On the event $E^C$, the adversary's random oracle calls are separate from all the oracle calls used to compute each $M_i$. Since random oracle calls are independent of all each other, the information from the adversary's random oracle calls is irrelevant toward determining if they are in world 0 or world 1. Therefore, unless $E$ occurs, the adversary is as good as an adversary with no random oracle calls. We complete the theorem by adding $\prob(E)$ to the advantage of an adversary with no random oracle calls.
    \begin{eqnarray*}
		\advantage &\leq& \mathbf{MaxAdv}_{\mathbf{0,\numqueries}} + \prob(E) \\
        &\leq&  {\numqueries\totaltime \over 2}
		\Bigl[ h^{-1} \Bigl(1 - {\alpha + \numprobes \over \keylength} \Bigr) \Bigr]^{\numprobes/2} +
		\numqueries\totaltime \cdot 2^{-\messagelength} +
		\frac{\numqueries}{\s+1}\bigg(\frac{4\messagelength\numqueries}{2^\messagelength} \bigg) ^\s + \frac{ \numqueries \numoraclecalls }{2^{\messagelength-1}} \,.
	\end{eqnarray*}
\bibliographystyle{unsrtnat}


\end{document}